\documentclass[11pt,reqno]{amsart}

\usepackage{lineno,hyperref,amsmath,amssymb,amsthm}
\usepackage{mathrsfs,graphicx,caption,subcaption}
\usepackage{amsfonts,enumerate}
\usepackage{cases,comment}

\usepackage{xcolor}
\modulolinenumbers[5]
\providecommand{\keywords}[1]
{\small\textbf{Keywords:} #1
}

\newtheorem{lemma}{Lemma}[section]

\newtheorem{proposition}{Proposition}[section]
\newtheorem{remark}{Remark}[section]

\newcommand\eps{\varepsilon}
\newcommand{\Section}[1]{\section{#1}\setcounter{equation}{0}}

\definecolor{lightgreen}{rgb}{.7,95,.65}

\newcommand{\R}{\mathbb{R}}
\newcommand{\Z}{\mathbb{Z}}

\newcommand{\overbar}[1]{\mkern 1.5mu\overline{\mkern-1.5mu#1\mkern-1.5mu}\mkern 1.5mu}

\author[Amadori]{Debora Amadori}
\address{Dipartimento di Ingegneria e Scienze dell'Informazione e Matematica (DISIM), University of L'Aquila -- L'Aquila, Italy}
\email{debora.amadori@univaq.it}
\author[Colangeli]{Matteo Colangeli}
\address{Dipartimento di Ingegneria e Scienze dell'Informazione e Matematica (DISIM), University of L'Aquila -- L'Aquila, Italy}
\email{matteo.colangeli1@univaq.it}
\author[Correa]{Astrid Correa}
\address{Dipartimento di Ingegneria e Scienze dell'Informazione e Matematica (DISIM), University of L'Aquila -- L'Aquila, Italy}
\email{astridherminia.correaluces@graduate.univaq.it}
\author[Rondoni]{Lamberto Rondoni}
\address{
Dipartimento di Scienze Matematiche, Politecnico di Torino, Corso Duca degli Abruzzi 24, 10129 Torino, Italy \\
and INFN, Sezione di Torino, Via P. Giuria 1, 10125 Torino, Italy\\
ORCID ID: 0000-0002-4223-6279}
\email{lamberto.rondoni@polito.it}

\begin{document}

\title[Exact Response and Kuramoto dynamics]{Exact Response Theory and Kuramoto dynamics}
\maketitle
\centerline{\date}






\begin{abstract}
The dynamics of Kuramoto oscillators is investigated in terms of the exact response theory based on the Dissipation Function,
which has been introduced in the field of nonequilibrium molecular dynamics. While linear response theory is a cornerstone of nonequilibrium statistical mechanics, it does not apply, 
in general, to systems undergoing phase transitions. Indeed, even a small perturbation may in that case result 
in a large modification of the state. An exact theory is instead expected to handle such situations. 
The Kuramoto dynamics, which undergoes synchronization transitions, is thus investigated analytically and numerically as a testbed for the exact theory mentioned above. 
A comparison between the two approaches shows how the linear theory fails, while the exact theory yields the correct response.
\end{abstract}

\\


\section{Introduction}
The response of a system with many degrees of freedom to an external stimulus is a central topic in nonequilibrium statistical mechanics. 
Its investigation has greatly progressed with the works of Callen, Green, Kubo, and Onsager, in particular, who contributed 
to the development of linear response theory \cite{Kubo,MPRV}.
In the '90s, the derivation of the Fluctuation Relations \cite{ECM,ES94,GC} provided the framework for a more general response theory, 
applicable to both Hamiltonian as well as dissipative deterministic particle systems \cite{MPRV,Caruso20,CL14,CR12,CRV12,dalron,ESW,Ruelle}. 
The study of response in stochastic processes, with a special focus on diffusion and Markov jump processes, has also been inspired by fluctuation 
relations, and has been studied {\em e.g.}\ in \cite{Agar,Maes,BDL,CMW11,Derrida}. 
Moreover, the role of causality, expressed by the Kramers-Kronig relations, 
in nonlinear extensions of the linear response theory has been discussed in \cite{CL12}.

The introduction of the Dissipation Function, first made explicit in \cite{ESAdvPhys}, and  developed as the observable of interest in 
Fluctuation Relations \cite{ESR2005,SRE2007}, paved the way to  an exact response theory. 
A theory expected to hold in presence of arbitrarily large perturbations and modifications of states, which allows the study
of the relaxation of particle systems to equilibrium or non-equilibrium steady states.

In this work we present and apply the Dissipation Function formalism to the Kuramoto model \cite{FirstKura,kuramoto1984chemical},   
which is considered a prototype of 
many particle systems exhibiting \textit{synchronization}, a phenomenon  familiar in many physical and biological contexts
\cite{fell2011role,Glass2001,jiruska2013synchronization,motter2013spontaneous,singer1999neuronal,STROGATZ20001}. 
Furthermore, the Kuramoto model provides the stage for a large research endeavor, in applied mathematics, control theory and statistical physics \cite{STROGATZ20001,RevModPhys.77.137,ARENAS200893,BALMFORTH200021,DORFLER20141539,gupta2018statistical,pikovsky2001}. See \cite{ha2016collective,DF-2018} for recent reviews on the subject.

In this paper, our aim is two-fold. On the one hand, we probe the exact response theory on a dissipative system with many degrees of freedom 
undergoing nonequilibrium phase transitions, which is in fact a challenging open problem.
On the other hand, while a vast mathematical literature exists on the Kuramoto model, it is interesting to analyze it from a new 
statistical mechanical perspective, in which some known results are reinterpreted, 
cf.\ {\em e.g.}\ Refs.\cite{BCM-CMS2015,DX-CMS2013}.

Our conclusion is that, while the linear response theory cannot characterize the Kuramoto synchronization process, the exact theory does. 
In particular, we obtain synchronization within the formalism of the Dissipation Function, thus showing how such a behaviour is captured by the
exact response theory, while it is not evidenced by the linear theory. Synchronization corresponds indeed to the maximum value of the
Dissipation Function, which we prove is attained in time. When the number of oscillators $N$ is large, this maximum value is proportional to the
oscillators coupling constant $K$.

This paper is organized as follows.
In Sec.~\ref{sec:sec1} we review some basic properties of the Kuramoto dynamics.
In Sec.~\ref{sec:sec2} we illustrate the main ingredients of the Dissipation Function response theory. 
In Sec.~\ref{sec:sec3} we study the response theory for the Kuramoto dynamics of identical oscillators. 
In Sec.~\ref{sec:sec4} we review the linear response theory, and we compare it with the exact response formalism. We draw our conclusions in Sec.~\ref{sec:sec5}.

\section{The Kuramoto system}
\label{sec:sec1}
The Kuramoto dynamics is defined on the $N$-dimensional torus, 
$\mathcal{T}^N = (\R/ (2\pi\Z))^N$, with $N \ge 1$, by the following set of coupled first order ODEs, for the phases $\theta_i(t)$:
\begin{equation}
\dot{\theta}_i=\omega_i+\frac{K}{N}\sum_{j=1}^N \sin(\theta_j-\theta_i) \qquad i=1,\dots,N  \label{kura1}
\end{equation}
where $K>0$ is a 
constant, and the natural frequencies $\omega_i\in \mathbb{R}$ are drawn from some given distribution $g(\omega)$.
The $N$ oscillators are represented by points rotating on the unit circle centered at the origin of the complex plane,
more precisely by $e^{i\theta_j}$ with $j=1,\ldots,N$. By introducing the polar coordinates of the barycenter,
\begin{equation}
R e^{i\Phi}=\frac{1}{N}\sum_{j=1}^N e^{i\theta_j} \label{order}
\end{equation}
with $R\in [0,1]$ and $\Phi\in \R$ (defined if $R>0$), one can rewrite Eq.(\ref{kura1}) as follows:
\begin{equation}
\dot{\theta}_i=\omega_i+K R \sin(\Phi-\theta_i) \, , \qquad i=1,\dots,N  \label{kura2}
\end{equation}
where $R=R(\theta(t))$ 
is the \textit{order parameter} and $\Phi=\Phi(\theta(t))$
the \textit{collective phase}, 
with $\theta=(\theta_1,...,\theta_N) \in \mathcal{M} = \mathcal{T}^N$, and $\mathcal M$ the phase space. 
The Kuramoto dynamics (\ref{kura2}) can also be written as a gradient flow:
\begin{equation}
\dot{\theta}=-\nabla f(\theta) \label{grad}
\end{equation}
with potential 
\begin{equation}
f(\theta)=-\sum_{i=1}^N \omega_i \theta_i +\frac{K}{2N}\sum_{i,j=1}^N \Big( 1-\cos(\theta_j-\theta_i) \Big) \, .
\end{equation}
that is analytic in $\theta$.

\vskip 5pt
\noindent
{\bf Identities for the order parameter.}\quad
{\em 
Equation \eqref{order} implies the following identities:
\begin{align}
R &= \frac 1 N \sum_{i=1}^N \cos(\Phi-\theta_i)\,, \label{ident-1}\\
0 &= \frac 1 N \sum_{i=1}^N \sin(\Phi-\theta_i)\,, \label{ident-1-bis}\\
R \sin(\Phi-\theta_i) & = \frac 1 N\sum_{j=1}^N \sin(\theta_j-\theta_i)\,,\qquad i=1,\ldots,N \label{ident-2bis}\\
R \cos(\Phi-\theta_i) & = \frac 1 N\sum_{j=1}^N \cos(\theta_j-\theta_i)
\,,\qquad i=1,\ldots,N\,. \label{ident-2}
\end{align}
Equations \eqref{ident-1} and \eqref{ident-2}, further imply:
\begin{equation}
R^2=\frac{1}{N^2}\sum_{i , j=1 }^N\cos(\theta_j-\theta_i) \,.
\label{R2}
\end{equation}
}

\vskip 5pt

\noindent
A \textit{complete frequency synchronization} occurs as $t\to+\infty$,  when  the differences
$\theta_i(t) - \theta_j(t)$ tend to a constant for all $i$ and $j$, and $R(\theta(t))$ tends to $R^\infty \in (0,1]$. 
Moreover, $R^\infty=1$ implies that all the $N$ terms of the sum in \eqref{ident-1} coincide with $\Phi$. In this case, the Kuramoto system undergoes a \textit{phase synchronization}.

\vskip 5pt
\noindent
For $\mathbf{\theta} \in \mathcal M$, 
we can rewrite 
Eq.(\ref{kura2}) as:
\begin{equation}
 \dot{\mathbf{\theta}}= W +
 V(\mathbf{\theta}) = V_K (\theta)
 \label{vfield}   
\end{equation}
where $W = \left( \omega_1 ,  \dots , \omega_N \right)$ is interpreted as an {\em equilibrium} vector field made of $N$ natural frequencies, 
while $V$ represents
a {\em nonequilibrium} vector perturbation with components:
\begin{equation}
V_i(\theta)= \frac{K}{N}\sum_{j=1}^N \sin(\theta_j-\theta_i) = 
K R \sin(\Phi-\theta_i) \, , \qquad i = 1, \dots , N\,.
\label{V-i}
\end{equation}
For later use, we prove the following identity.
\begin{lemma}\label{lem:2}
The divergence of the Kuramoto vector field $V_K$ of Eq.\eqref{vfield}, {\em i.e.}\ the associated phase space volumes variation rate $\Lambda$,
satisfies: 
\begin{equation}
\Lambda :=
\text{\rm div}_\theta V =  
K \left( 1 - N R^2 \right)
\,. \label{div}
\end{equation}
\begin{proof}
By means of \eqref{V-i}, for $i=1,\ldots,N$ one has
\begin{align*}
    \partial_{\theta_i} V_i 
    &= \frac{K}{N}  \partial_{\theta_i}\left( \sum_{i\not=j=1}^N \sin(\theta_j-\theta_i)\right)  \\
    &= - \frac{K}{N} \left( \sum_{i\not=j=1}^N \cos(\theta_j-\theta_i)\right) = - \frac{K}{N} \left( \sum_{j=1}^N \cos(\theta_j-\theta_i) -1 \right)\\
    &=  - {K} R \cos(\Phi-\theta_i)    + \frac{K}{N}    
\end{align*}
where we used \eqref{ident-2}. Summing  over $i$, and using \eqref{ident-1}\,, Eq.\eqref{div}  follows.
\end{proof}
\end{lemma}

\par\noindent
Therefore, the Kuramoto dynamics do not preserve the phase space volumes, and $\Lambda$ actually varies in time, 
since $R$ is a function of the dynamical variables $\theta(t)$.

\Section{Mathematical framework of Response theory}
\label{sec:sec2}
Let us summarize the mathematical framework of the exact response theory originally derived in Ref.\cite{ESW}, 
and further developed in {\em e.g.}\ Refs.\cite{Caruso20,ESW,SRE2007,Typic,Jepps16}. 
The starting point is a flow $S^t: \mathcal{M}\rightarrow \mathcal{M}$, with phase space $\mathcal{M} \subset \mathbb{R}^N$, $N\ge1$, 
that is usually determined by an ODE system 
\begin{equation}
\dot{\theta} = V(\theta) \, , \quad 
\theta \in \mathcal{M}
\label{ODEgen}
\end{equation}
with $V$ a vector field on $\mathcal M$. Let $S^t \theta$ denote the solution at time $t\in \R$, with initial condition $\theta$, of such ODEs. 
The second ingredient is a probability measure $d \mu_0(\theta) = f_0(\theta) d \theta$ on  $\mathcal{M}$, with positive and continuously
differentiable density $f_0$. 
A time evolution is induced on the simplex of probabilities on $\mathcal M$, defining the probability at a time $t\in\R$ as:
\begin{equation*}
    \mu_t(E) = \mu_0(S^{-t}E)
\end{equation*} 
for each measurable set $E\subset \mathcal M$. 
This amounts to consider probability in a phase space like the mass of a fluid in real space. 
The corresponding continuity equation for the probability densities is the (generalized)
Liouville equation:
\begin{equation}
\frac{\partial f}{\partial t} +\text{div}_\theta(f V) =0 \, .
\label{liouv}
\end{equation}
Denoting by $f_t$ the solution of Eq.\eqref{liouv} with initial datum $f_0$, we can write $d \mu_t = f_t d \theta$.  
Letting $\Lambda = \mbox{div}_\theta \, V$ be the phase space volumes variation rate, and
introducing the {\em Dissipation Function} $\Omega^{f,V}$ \cite{SRE2007,Jepps16}:
\begin{equation}
\Omega^{f,V}(\theta):=-\Lambda(\theta) - V(\theta)\cdot \nabla \log f(\theta) \, , ~~ 
\nabla = \left( \partial_{\theta_1} , ...  \, , \partial_{\theta_N} \right) \, 
\label{Omega}
\end{equation}
the Euler version
of the Liouville equation \eqref{liouv} may be written as:
\begin{equation}
\frac{\partial f}{\partial t}= f ~\Omega^{f,V} \, .
\label{dyn}
\end{equation}
which can also be cast in the Lagrangian form:
\begin{equation}
\frac{d f}{d t}= - f \, \Lambda \, , 
\label{dyn0}
\end{equation}
with $\frac{d}{d t} = \frac\partial{\partial t} + V \cdot \nabla_\theta$ the total derivative along the flow \eqref{ODEgen}.

Direct integration of Eq.(\ref{dyn0}) yields
\begin{equation}
    f_{s+t}(S^t \theta)=\exp\{-\Lambda_{0,t}(\theta)\} f_s(\theta) \, ,
    \quad \forall\, t,\, s\ge 0 \label{dyn3}
\end{equation}
where we used the notation
\begin{equation}
    \mathcal{O}_{s,t} (\theta) :=
    \int_s^t \, \mathcal{O} (S^\tau \theta) d \tau
    \label{notat}
\end{equation}
for the phase functions, or \textit{observables}, $\mathcal{O}:\mathcal{M}\rightarrow \mathbb{R}$, so that, in particular, 
${\Lambda}_{0,t}(\theta)= \int_{0}^t \Lambda(S^{\tau}\theta) d\tau$. 

In the following Proposition, this notation is used with the observable $\mathcal{O}= \Omega^{f,V}$, so that the time integral in \eqref{notat} will correspondingly be denoted by $\Omega^{f,V}_{s,t}$.

\begin{proposition}\label{sec3:prop1}
For all $t$, $s\in\R$, the following identity holds:
\begin{equation}
    f_{s+t}(\theta)=
    \exp\Big\{\Omega^{f_s,V}_{-t,0}(\theta)\Big\}f_s(\theta)\,. \label{ft}
 \end{equation}
\end{proposition}

\begin{proof} We start by claiming that 
\begin{equation}
\Omega^{f_t,V}_{0,s}(\theta)=
    \log \frac{f_t(\theta)}{f_t(S^{s} \theta)}-\Lambda_{0,s}(\theta)\,. \label{omevol}
\end{equation}
Indeed, one has:
\begin{equation}
 V(S^u \theta)\cdot \nabla \log f_t (S^u \theta)=\frac{d}{du}\log f_t (S^u \theta) \label{du}
\end{equation}
 because $t$ is fixed and $f_t$ does not depend explicitly on $u$, hence Eqs.(\ref{Omega}) and (\ref{du}) imply:
 \begin{align*}
 \Omega^{f_t,V}_{0,s}(\theta)&=
   - \int_0^s \big[ \Lambda(S^u\theta)+ V\cdot \nabla \log f_t(S^u\theta) \big]  du \\
  &= -\Lambda_{0,s}(\theta)-\int_0^s \frac{d}{du}\log f_t (S^u \theta) d\theta=-\Lambda_{0,s}(\theta)-\log \frac{f_t(S^{s} \theta)}{f_t(\theta)}
 \end{align*}
which leads to Eq.(\ref{omevol}). Next, Eqs.(\ref{dyn3}) and Eq.(\ref{omevol}) yield
\begin{equation}
\exp\Big\{\Omega^{f_s,V}_{s,s+t}(\theta)\Big\} f_{s}(S^{s+t} \theta)=     
     \exp\Big\{-\Lambda_{s,s+t}(\theta)\Big\}f_{s}(S^{s} \theta)=f_{s+t}(S^{s+t} \theta) \label{omevol2}
 \end{equation}
 which produces \eqref{ft}.
\end{proof}

As a consequence of Proposition~\ref{sec3:prop1}, 
a probability density $f$ is {\it invariant} under the dynamics 
if and only if $\Omega^{f,V}$ identically vanishes:
\begin{equation}\label{steady}
\Omega^{f,V}(\theta) = 0\, , \quad \forall~ \theta \in \mathcal{M}\,.
\end{equation}
In the sequel, we shall use the notation
\begin{equation}
\langle \mathcal{O}\rangle_t :=\int_{\mathcal{M}}  \mathcal{O}(\theta) f_t(\theta) d\theta \label{aver}
\end{equation} 
to denote the average of an observable with respect to the probability measure $\mu_t = f_t \, d \theta$. 
The exact response theory based on the Dissipation Function states that the average $\langle \mathcal{O}\rangle_t$ can be expressed in terms of the known initial density $f_0$, as in linear response theory. The difference between the two theories lies in the correlation functions that must be integrated in time. 
\begin{lemma} {\bf (Exact response):}
\label{lem:1}
Given $\{ S^t \}_{t \in \mathbb{R}}$ and an integrable observable $\mathcal{O} : \mathcal{M} \to \mathbb{R}$,  the following identity holds:
\begin{equation}
\langle \mathcal{O}\rangle_t= \langle \mathcal{O}\rangle_0+\int_0^t \langle (\mathcal{O} \circ S^\tau)\ \Omega^{f_0,V}\rangle_0 \ d\tau\,. \label{response}
\end{equation}
\end{lemma}

\begin{proof}
First of all, $f_0$ is smooth as a function of $\theta$ by assumption, and evolves according to the Liouville equation. 
Therefore, $f_t$ is also smooth with respect to $\theta$ and $t$ for every finite time $t$. 
In turn, $\Omega^{f_t,V}(\theta)$ is differentiable with respect to $\theta$ and $t$, if $f_0$ (that depends only on $\theta$) is differentiable with respect to $\theta$. 
These conditions are immediately verified for differentiable $f_0$, and smooth dynamics on a compact manifold. Therefore two identities can be derived for integrable $\mathcal{O}$:
\begin{eqnarray*}
 \mathcal{O}_{0,s}(\theta)&=&\int_{0}^s\mathcal{O}(S^{u}\theta) du= \int_{\tau}^{s+\tau}\mathcal{O}(S^{u-\tau}\theta) du=\int_{\tau}^{s+\tau}\mathcal{O}(S^{-\tau}S^{u}\theta) du\nonumber \\
 &=&\mathcal{O}_{\tau,s+\tau}(S^{-\tau}\theta) 
\end{eqnarray*}
which is valid for every $\tau\in\R$, and
\begin{eqnarray}
   \langle \mathcal{O}\rangle_{t+s}&=&\int \mathcal{O}(\theta) f_{t+s}(\theta) d\theta\nonumber\\
   &=&\int \mathcal{O}(S^s(S^{-s}\theta)) f_{t+s}(S^s(S^{-s}\theta)) \left|\frac{\partial \theta}{\partial (S^{-s}\theta)}\right| d(S^{-s}\theta) \nonumber \\
   &=& \int \mathcal{O}(S^s(S^{-s}\theta)) f_{t+s}(S^s(S^{-s}\theta)) \exp\Big\{\Lambda_{-s,0}(\theta)\Big\} d(S^{-s}\theta)\nonumber \\
   &=& \int \mathcal{O}(S^s(S^{-s}\theta)) f_{t+s}(S^s(S^{-s}\theta)) \exp\Big\{\Lambda_{0,s}(S^{-s}\theta)\Big\} d(S^{-s}\theta)\nonumber \\ 
   &=&\int \mathcal{O}(S^s\theta) f_{t+s}(S^s\theta) \exp\Big\{\Lambda_{0,s}(\theta)\Big\} d\theta =\int \mathcal{O}(S^s\theta) f_{t}(\theta) d\theta \nonumber \\
   &=& \langle \mathcal{O}\circ S^s \rangle_t 
   \label{idB}
\end{eqnarray}
to obtain \cite{Jepps16}:
\begin{equation}
    \frac{d}{ds}\langle\mathcal{O}\rangle_{s}=\langle \mathcal{O} \ (\Omega^{f_r,V}\circ S^{r-s})\rangle_s \label{ds1}
\end{equation}
which holds $\forall r\ge 0$.
Note that in Eq. \eqref{idB} we used the relation
\begin{equation}
\label{idb2}
   \left|\frac{\partial \theta}{\partial (S^{-s}\theta)}\right|= \exp\Big\{\Lambda_{-s,0}(\theta) \Big\} 
\end{equation}
which is discussed in \ref{app:appB}, see Eq. \eqref{jac2}. Choosing $r=0$ in \eqref{ds1}, one finds
\begin{equation}
    \frac{d}{ds}\langle\mathcal{O}\rangle_{s}=\langle \mathcal{O} \ (\Omega^{f_0,V}\circ S^{-s})\rangle_s=\langle (\mathcal{O}\circ S^{s}) \ \Omega^{f_0,V}\rangle_0 \label{ds2}
\end{equation}
where we used (\ref{idB}).
Then, integrating over time from $0$ to $t$, Eq.(\ref{ds2}) yields (\ref{response}).
\end{proof}

The apparently peculiar definition of the Dissipation Function is motivated by the fact that it can be associated with the energy dissipation of particle systems, if $f_0$ is properly chosen.
In particular, this is the case for models of nonequilibrium molecular dynamics, such as the Gaussian and the Nos\'e - Hoover thermostatted systems, if $f_0$ is the invariant probability 
density for the corresponding equilibrium dynamics, 
{\em i.e.}\ the dynamics subjected to the same constraints of the nonequilibrium ones, in which the dissipative forces are switched off. In other words, $\Omega^{f_0,V}$ equals the energy dissipation if 
$\Omega^{f_0,V_0} \equiv 0$ and $V_0$ is the (non dissipative) vector field implementing the same constraints that $V$ does \cite{SRE2007}. Typical constraints are the constant internal energy, the constant kinetic energy, the constant temperature, the constant pressure etc.. The state characterized by $f_0$ may be prepared like that at start. Alternatively, one usually thinks that it is generated by the equilibrium dynamics:
\begin{equation}
\dot{\mathbf{\theta}}=V_0(\mathbf{\theta}) \label{refer}
\end{equation} 
started long before the time $t=0$, so that at time 0 it is realized. 
While this is not mathematically required, it is physically convenient, and it helps our intuition to assume that $\mu_0$ is invariant under the dynamics \eqref{refer}, which we call  \textit{unperturbed} or \textit{reference} dynamics.
At time $t=0$, the dynamics \eqref{refer} is perturbed and the perturbation remains in place for all $t > 0$. 

In general, the density $f_0$ is not invariant under the perturbed vector field $V$, cf.\ Eq.\eqref{ODEgen}.
Therefore, it will evolve as prescribed by Eq.(\ref{dyn}) into a different density, $f_t$, at time $t>0$. Nevertheless, Eq.(\ref{response}) expresses the average $\langle \mathcal{O}\rangle_t $ 
in terms of a correlation function computed with respect to $f_0$, the non-invariant density, which is only invariant under the unperturbed dynamics.

The full range of applicability of this theory is still to be identified. However, it obviously applies to smooth dynamics on smooth compact manifolds, such as the Kuramoto dynamics (\ref{kura1}), 
which has $\mathcal{M}= \mathcal{T}^N$. One advantage of using the Dissipation Function, compared to other possible exact approaches to response, apart from molecular dynamics efficiency, 
is that $\Omega^{f_0,V}$ corresponds to a physically measurable quantity, {\em e.g.}\ proportional to a  current, that is adapted to the initial state of the system of interest. 
Moreover, it provides necessary and sufficient conditions for relaxation of ensembles, as well as sufficient conditions for the single system relaxation, known as T-mixing \cite{SRE2007,Jepps16}.
The analysis of the response theory for a specific example of the Kuramoto model is discussed in the next Section.

\Section{Response theory for identical oscillators}
\label{sec:sec3}
Let us focus on the case of identical oscillators, namely the Kuramoto dynamics in which all the natural frequencies $\omega_i$ in Eq.\eqref{kura1} equal the same constant $\omega \in \mathbb{R}$.
In particular, let the unperturbed dynamics be defined by the vector field $V_0(\theta)=W=(\omega, \dots , \omega)$,
which corresponds to $K=0$ in Eq.\eqref{kura1}, {\em i.e.}\ to decoupled oscillators, equipped with same natural frequency.
Such dynamics are conservative, since
${\rm div}_\theta V_0 = 0$. The corresponding steady state can then be considered an equilibrium state. At time $t=0$ the perturbation $V$ is switched on,
and we can write:
\begin{equation}\label{defV}
\dot{\mathbf{\theta}}=\begin{cases}
W & t<0\\
W + V(\theta) & t > 0 \, .
\end{cases}
\end{equation} 
The perturbed dynamics corresponds to the Kuramoto dynamics \eqref{kura1}, which is not conservative, cf.\ Eq.\eqref{div}. 
As an initial probability density, invariant under the unperturbed dynamics, we may take the factorized density:
\begin{equation}
f_0(\theta)=(2 \pi)^{-N}     \label{invden}
\end{equation}
which, indeed, yields:
\begin{equation}\label{invf0}
\Omega^{f_0,V_0}=-\left(\text{div}V_0+V_0\cdot \nabla \log f_0 \right)\equiv 0 \, , \quad
\mbox{and ~} \quad  {\frac{\partial f} {\partial t}} = 0\,.
\end{equation}
After the perturbation, the Dissipation Function takes the form:
\begin{equation}
\Omega^{f_0,V}= - \left( \text{div}_\theta V + V \cdot \nabla \log f_0 \right) = K \left( NR^2 - 1 \right) = 
\frac{K}{N}\sum_{i , j=1 }^N\cos(\theta_j-\theta_i) - K
\label{omega}
\end{equation}
and the density evolves as:
\begin{equation}
    f_t(\theta)= \frac{1}{(2\pi)^N} 
\exp \Big[ -K \left( t - N R_{-t,0}^2(\theta) \right) \Big] 
\label{DensEvol}
\end{equation}
where $R_{-t,0}$ denotes the integral of $R$ from time $-t$ to 0, cf.\ Eq.\eqref{notat}.

\begin{remark}\label{cinf}
The Dissipation Function Eq.\eqref{omega} is of class $C^\infty$.
\end{remark}

\noindent
Using the formula (\ref{response}) to compute the response for the observable $\mathcal{O}=\Omega^{f_0,V}$, we obtain:
\begin{equation}
\langle \Omega^{f_0,V}\rangle_t= \langle \Omega^{f_0,V}\rangle_0 + \int_0^t \langle (\Omega^{f_0,V}\circ S^{\tau}) \Omega^{f_0,V}\rangle_0 \ d\tau 
\label{response2}
\end{equation}
that is 
\begin{align*}
    &\int_{\mathcal{M}}  \Omega^{f_0,V}(\theta) f_t(\theta) d\theta \\
    &\qquad  = (2 \pi)^{-N} \int_{\mathcal{M}}  \Omega^{f_0,V}(\theta) d\theta
    + (2 \pi)^{-N} \int_0^t \int_{\mathcal{M}}  \Omega^{f_0,V} (S^\tau(\theta)) \Omega^{f_0,V} (\theta) \, d\theta d\tau\,.
\end{align*}
Moreover:
\begin{equation}
 \langle R^2 \rangle_0=\frac{1}{N}  \, ,
\quad \mbox{hence } \, \, \,
\langle \Omega^{f_0,V} \rangle_0=K \big( N \langle R^2 \rangle_0 - 1 \big) = 0 
\label{ident2-bis}
\end{equation}
as expected.
\begin{remark}
Note that the scalar field $\Omega^{f_0,V_0}$ is identically 0, while $\Omega^{f_0,V}$ is not, see Eq.\eqref{omega}. However, the phase space average  $\left\langle \Omega^{f_0,V} \right\rangle_0$ vanishes.
\end{remark}
\noindent
Therefore, using Eqs.\eqref{response} and \eqref{omega}  we can write:
\begin{align*}
    \left\langle \Omega^{f_0,V}\right\rangle_t &=  \int_0^t \left\langle (\Omega^{f_0,V}\circ S^{\tau}) \Omega^{f_0,V}\right\rangle_0 \ d\tau \\
    &= KN \int_0^t \left\langle \Omega^{f_0,V} \left[ R^2 \circ S^{\tau} \right] \right\rangle_{0} d\tau - K \int_{0}^t \left\langle \Omega^{f_0,V} \right\rangle_0 \, d\tau\\ 
    &= KN \int_{0}^t \left\langle \Omega^{f_0,V} \left[ R^2 \circ S^\tau \right] \right\rangle_0\, d\tau \\
    &= K^2 N^2 \int_{0}^t \left\langle R^2 \left[ \, R^2 \circ S^\tau \right] \right\rangle_0 \,d \tau - K^2 N \int_{0}^t \left\langle R^2 \circ S^\tau \right\rangle_0\, d\tau \, .
\end{align*}
For the second integral we have:
\begin{align*}
    \int_0^t \left\langle R^2 \circ S^\tau \right\rangle_0 d\tau &= \frac{1}{(2\pi)^N} \int_0^{t} \int_{\mathcal{M}} R^2 (S^\tau \theta) d\theta d\tau \\
    &= \frac{1}{(2\pi)^N} \int_0^{t} \int_{\mathcal{M}} R^2 (S^\tau \theta) \left| \frac{\partial\theta}{\partial S^\tau \theta} \right| dS^\tau \theta d\tau \\
&= \frac{1}{(2\pi)^N} \int_0^{t} \int_{\mathcal{M}} R^2 (S^\tau \theta)  \exp\Big\{\Lambda_{0,\tau}(\theta)\Big\} dS^\tau \theta \, .
\end{align*}
Explicit calculations can be carried out for $N=2$ and will be discussed in Sec. \ref{subsec:N2}, while the study of the general case with $N>2$ is deferred to Sec. \ref{subsec:generalN}.

\subsection{The case with two oscillators}
\label{subsec:N2} 

\noindent
For $N=2$ and $\omega\ge 0$, consider the system for two oscillators: 
\begin{equation}\label{eq:N=2-nonid}
    \begin{cases}\displaystyle
    \dot{\theta}_1 =  \frac {\omega}2 
    + \frac{K}{2} \sin(\theta_2 - \theta_1) &
    \\[2mm]
    \displaystyle
    \dot{\theta}_2 = - \frac {\omega}2 
    + \frac{K}{2} \sin( \theta_1 - \theta_2)\,. &
    \end{cases}
\end{equation}
In the case in which all natural frequencies coincide, as in Eq.\eqref{eq:N=2-nonid} for $\omega=0$, the oscillators are referred to as \textit{identical}. Setting 
$\psi= \theta_1 - \theta_2$, we obtain the %
following equation:
\begin{equation}
    \frac{d \psi}{dt}= \omega - K\sin(\psi) \label{psi}\,.
\end{equation}
With a slight abuse of notation, in the following we denote by $S^t\theta$, $S^t\psi$ the flows corresponding to \eqref{eq:N=2-nonid}, \eqref{psi} respectively, with initial data $\theta=(\theta_1,\theta_2)$ and $\psi=\theta_1-\theta_2$. Then, the solution of \eqref{psi} can be explicity expressed as 
\begin{equation}
\tan\left(\frac{S^t\psi}2\right) 
= g(\psi,t) 
\label{g}
   \end{equation}
where:
\begin{itemize}
\item if $K > \omega= 0$, then
\begin{equation*}
    g(\psi,t)= e^{-Kt} \tan\left(\frac{\psi}{2}\right) \,;
\end{equation*}    

 \item if $K > \omega>0$, then
   \begin{align*}
       g(\psi,t) &=  \frac{K}{\omega} + \frac{\sqrt{K^2 - \omega^2}}{\omega}  \cdot
\frac{1+h_1(\psi)\, e^{t\sqrt{K^2 -\omega^2}}} 
{1-h_1(\psi)\, e^{t\sqrt{K^2 -\omega^2}}}       
       \\[2mm]
       h_1(\psi) &= \frac{\omega\tan(\frac\psi 2)-K - \sqrt{K^2 -\omega^2}}{\omega\tan(\frac\psi 2)-K + \sqrt{K^2 -\omega^2}}
           \,.
    \end{align*}
 
The formulas here above can be deduced by \cite[Lemma D.2]{CHOI2012735}, Case 1;

\item if $0\le K< \omega $, then 
\begin{align*}
    g(\psi,t) &= \frac{K}{\omega} + \frac{\sqrt{\omega^2 - K^2}}{\omega} \tan \left( \frac{t\sqrt{\omega^2 - K^2} }{2} + h_2(\psi)  \right)  \\
    h_2(\psi) &= \arctan \frac{\omega \tan \left( \frac{\psi}{2} \right) - K}{\sqrt{\omega^2 - K^2}}\,,
\end{align*}
see \cite[Lemma D.2]{CHOI2012735}, Case 3 with $R^\infty=\omega/K$. 
\end{itemize}

\noindent
Recalling Eq.\eqref{R2} and using the identity $1+\cos x = 2 \left(1+ \tan^2\left(\frac x2\right) \right)^{-1}$,  
we find that $(R^2 \circ S^t)$ 
can be written as
\begin{equation}
   R^2(S^t \theta)
   = \frac{1}{2} \left[ 1+ \cos(S^t\psi) \right] = \frac{1}{g^2(S^t \psi)+1}\, , \label{R2g}
\end{equation}
For $\omega=0$, one explicitly obtains:
\begin{equation}
    R^2(S^t \theta) = \left(\tan^2 \left(\frac{ \psi}
    {2} \right)e^{-2K t}+1 \right)^{-1} \label{R2g1}
\end{equation}
and
\begin{eqnarray}
&&S^t \psi \to 0\quad \mbox{for } t \to +\infty \, ,\qquad \mbox{if}\ |\psi| \ne \pi
\\
&&|S^t  \psi| \to \pi\quad \mbox{for } t \to - \infty \, ,\qquad \mbox{if}\ 
\psi \ne 0\,.
\end{eqnarray}
In particular, for $\theta_1 \ne \theta_2$ 
and $\theta_1$, $\theta_2\in [0,2\pi)$,
the $t \to -\infty$ limit yields $S^t \psi \to -\pi$ if $\theta_1 < \pi$, and $S^t \psi \to \pi$ if $\theta_1 > \pi$. 
Then, the set 
$$
E_\infty=\{ (\theta_1,\theta_2) \in \mathcal{T}^2 : \theta_1 = \theta_2 \}
$$ 
is invariant and attracting for the Kuramoto dynamics, while the set 
$$
E_{-\infty}=\{ (\theta_1,\theta_2) \in \mathcal{T}^2 : |\theta_1 - \theta_2| = \pi \}
$$
is invariant and repelling. This also implies that: 
$$
R^2(S^t \theta) \to 0 \, ,\quad 
\Omega^{f_0,V} \to -K  \, ,\qquad 
\mbox {for } ~~ \psi \ne 0 \, , ~
t \to - \infty
$$
while 
\begin{equation}
    R^2(S^t \theta) \to 1 \, ,\quad 
    \Omega^{f_0,V} \to K \, ,\qquad  
    \mbox{for } ~~ |\psi| \ne \pi \, , ~t \to \infty \, .
\nonumber
\end{equation}
\vskip 4pt \noindent
Consequently, Eq.\eqref{DensEvol} shows that the probability piles up on the zero Lebesgue measure sets $E_\infty$ and $E_{-\infty}$, respectively for $t \to \infty$ and $t \to -\infty$.

For $\tau \ge 0$, the following relations also hold:
\begin{equation}
\label{r2}
    \left\langle R^2 \circ S^{\tau} \right\rangle_0 = \frac{1}{(2\pi)^2} \int_{\mathcal{M}} \frac{1}{\tan^2(\frac{\theta_1 - \theta_2}{2})e^{-2K\tau}+1} d\theta = 
\frac{1}{e^{-K\tau}+1}
\end{equation}
and 
\begin{equation*}
    \left\langle R^2 (R^2 \circ S^\tau) \right\rangle_0 = \frac{1}{8\pi^2} \int_{\mathcal{M}} \frac{1+\cos(\theta_1 - \theta_2)}{\tan^2(\frac{\theta_1 - \theta_2}{2})e^{-2K\tau}+1} d\theta =
\frac{2e^{-K\tau}+1}{2(e^{-K\tau}+1)^2}
\end{equation*}
which then yields
\begin{equation*}
    \int_{0}^t \left\langle  R^2 \circ S^\tau \right\rangle_0 \,d \tau = 
    t + \frac{ \ln \left( e^{-Kt}+1 \right)}{K} - \frac{\ln(2)}{K}
\end{equation*}
and
\begin{align*}
    & \int_{0}^t \left\langle R^2 (R^2 \circ S^\tau)  \right\rangle_0\, d\tau \\
    & \qquad = \frac{t}{2} + \frac{1}{2K} \left[ \frac{3}{2} + \ln \left( \frac{e^{-Kt}+1}{2} \right) - \frac{2}{e^{Kt}+1} - \frac{1}{e^{-Kt}+1} \right] \,.
\end{align*}
Thus, we finally obtain the explicit expressions
\begin{equation}
    \left\langle \Omega^{f_0,V} \right\rangle_t= 
    K \tanh \left(\frac{K t}{2}\right) \label{Omexpl}
\end{equation}
and
\begin{equation}
    \left\langle (\Omega^{f_0,V}\circ S^t) \Omega^{f_0,V} \right\rangle_0= 
    \frac{K^2}{1+\cosh(Kt)} \,.
    \label{omom}
\end{equation}
In the limit $t\to+\infty$, we thus find the asymptotic values
\begin{equation}\label{eqdiss}
    \left\langle \Omega^{f_0,V} \right\rangle_t 
    \to K    \qquad \text{and} \qquad \left\langle (\Omega^{f_0,V}\circ S^t) \Omega^{f_0,V} \right\rangle_0 \to 0
\end{equation}
In particular, the two-time autocorrelation of $\Omega^{f_0,V}$ is monotonic
as also shown in the two panels of Fig.\ref{fig:fig1}. Indeed, Eq.\eqref{omom} yields, for $t\ge 0$:
\begin{equation*}
    {\frac{\rm d}{{\rm d} t}} 
    \left\langle (\Omega^{f_0,V}\circ S^t) \Omega^{f_0,V} \right\rangle_0 = 
-K^2 {  
\frac{\sinh K t} {\left(1 + \cosh K t \right)^2}}\le 0 \, .
\end{equation*}

\subsection{The general case} 
\label{subsec:generalN}
In this Subsection we assume $N\ge2$ and $\omega=0$, considering the following dynamics:
\begin{equation}\label{kura1-bis}
\dot{\theta}_i=\frac{K}{N}\sum_{j=1}^N \sin(\theta_j-\theta_i) = K R \sin(\Phi-\theta_i) \, , \qquad i=1,\dots,N \,. 
\end{equation}
where $R$ and $\Phi$ are 
defined in Eq.\eqref{order}.
We are going to prove that the observable $\left\langle\Omega^{f_0,V}\right\rangle_t$ is a monotonic function of time, and we can estimate the asymptotic value it attains in the large time limit.
\begin{figure}
    \centering
    \includegraphics[width=0.45\textwidth]{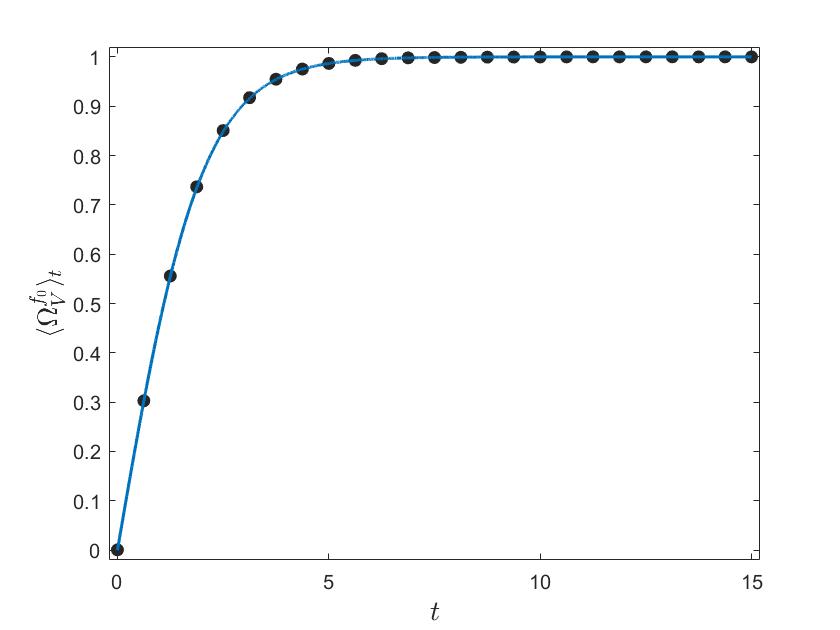}
    \includegraphics[width=0.45\textwidth]{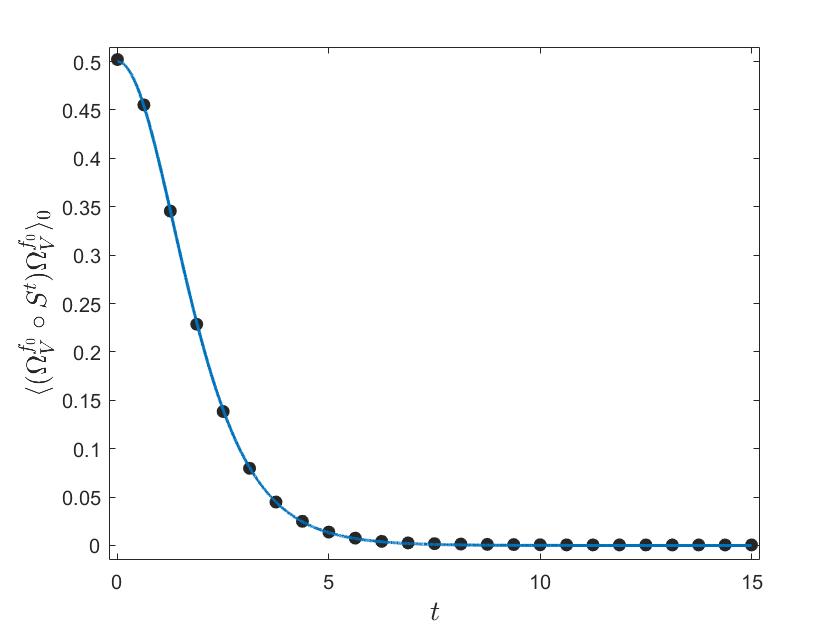}
    \caption{Behavior of $\langle\Omega^{f_0,V}\rangle_t$ and $\langle(\Omega^{f_0,V}\circ S^t)\Omega^{f_0,V}\rangle_0$ as functions of time, for $N=2$, $K=1$ and $\omega=0$. Disks and solid lines correspond to the numerical and analytical results, respectively. The averages were taken over a set of $5000$ trajectories with initial data sampled from the uniform distribution on $[0,2\pi)$.}
    \label{fig:fig1}
\end{figure}
We start by proving the following result.
\begin{lemma}\label{lem:5}
For every $t>0$, the time derivative of the expectation of the Dissipation Function obeys:
\begin{equation}
\label{ineq2}
\frac{d}{dt}\left(\Omega^{f_0,V}(S^t \theta)\right) \ge 0 \, ~~ \mbox{ and } ~ ~
    \frac{d}{dt}\left\langle\Omega^{f_0,V}\right\rangle_t=\left\langle(\Omega^{f_0,V}\circ S^t)\Omega^{f_0,V}\right\rangle_0 \ge 0 \,.
\end{equation}
\begin{proof}
First, we note that by setting $\mathcal{O}=\Omega^{f_0,V}$ in Eq. \eqref{ds2}, we find:
\begin{equation}
    \frac{d}{dt}\langle\Omega^{f_0,V}\rangle_{t}=\langle (\Omega^{f_0,V}\circ S^{t}) \ \Omega^{f_0,V}\rangle_0 \,. \label{res1}
\end{equation}
Moreover, Eq. \eqref{idB} with $t=0$ and $\mathcal{O}=\Omega^{f_0,V}$ yields:
\begin{equation}
  \left\langle\Omega^{f_0,V} \right\rangle_t=\left\langle \Omega^{f_0,V}\circ S^t\right\rangle_0\,. \label{ident}
\end{equation}
Therefore, we can write:
\begin{eqnarray}\label{ident2}
    \frac{d}{dt}\left\langle\Omega^{f_0,V}\circ S^t \right\rangle_0&=&\frac{d}{dt}\int_\mathcal{M}\Omega^{f_0,V}(S^t \theta)f_0(\theta) d\theta \nonumber\\
    &&\hskip -50pt
    =\int_\mathcal{M}\frac{d}{dt}\left(\Omega^{f_0,V}(S^t \theta)\right)f_0(\theta) d\theta=\left\langle \frac{d}{dt}\left(\Omega^{f_0,V}(S^t \theta)\right)\right\rangle_0 
\end{eqnarray}
Then, using Eq.(2.5) in Ref.\cite{BCM-CMS2015} we find:
\begin{equation}\label{eq:deriv-R2}
 \frac{d}{dt} R^2(S^t \theta) = \frac{2 K}{N}R^2(S^t \theta) \sum_{j=1}^N \sin^2\left(S^t\theta_j-\Phi\left(S^t \theta\right)\right) 
\end{equation}
where $S^t \theta_j$ denotes the $j-$th element of $S^t \theta$, and then 
\begin{equation}\label{ineq1}
    \frac{d}{dt}\left(\Omega^{f_0,V}(S^t \theta)\right)=2 K^2 R^2(S^t \theta) 
    \left[\sum_{j=1}^N \sin^2\left(S^t\theta_j-\Phi\left(S^t \theta\right)\right)  \right]\ge 0
\end{equation}
for all $\theta\in \mathcal{M}$. By integrating over $\mathcal{M}$ we obtain \eqref{ineq2}. This completes the proof.
\end{proof}
\end{lemma}

\begin{remark}
Unlike stationary current autocorrelations, that may fluctuate between positive and negative values, 
the two-time autocorrelation of $\Omega^{f_0,V}$, computed with respect to the initial probability measure, is non-negative.
\end{remark}

Theorem 2.4 of Ref.\cite{BCM-CMS2015} shows that non stationary solutions of the system \eqref{kura1-bis} converge, as 
$t\to+\infty$, 
either to a complete 
frequency synchronized state $\Theta^*$,
{\em i.e.}\ to a state denoted by $(N,0)$, that takes the form:
\begin{equation}
\Theta^*=\left(\varphi^*, \ldots, \varphi^*\right)  
\end{equation}
in which all phases are equal; or to a state denoted by $(N-1,1)$, that takes the form:
\begin{equation}
\Theta^\dagger =\left(\varphi^* + k_1\pi,
\varphi^*+k_2\pi,\varphi^*+k_3\pi,\varphi^*+k_4\pi,\ldots,\varphi^*+k_N\pi\right) 
\end{equation}
 where $k_i \in \{-1,+1\}$ 
 for a single $i \in \{1,2,...,N\}$, and all $k_j=0$ with $j \ne i$. 
 This can be understood also in terms of the Dissipation Function. In the first place, without loss of generality, let us consider a fixed point $\bar \theta$ of type $(N-1,1)$ whose antipodal is in the $N$-component, {\em i.e.}\ 
\begin{equation} 
 \bar \theta=(\varphi^*,\ldots,\varphi^* ,(\varphi^* +\pi) \, \mbox{mod} \, 2 \pi)   \label{eq:SType2}
\end{equation}
for a $\varphi^* \in [0,2\pi)$. Then, the following holds:

\begin{proposition}\label{SetNPrepo}
The set of initial data such that the solution to \eqref{kura1-bis} reaches a stationary  $(N-1,1)$-state for $t \to + \infty$ has 0-measure. 
\begin{proof}
For $V(\theta)$ as in \eqref{vfield}, the
Jacobian matrix $A(\theta) \dot = \nabla V (\theta)$ is given by 
\begin{equation}\nonumber
A_{ij}=
    \begin{cases}
    \frac{\partial V_j}{\partial \theta_i}=\frac{1}{N} \cos(\theta_i-\theta_j), & 
    i \neq j   \\ \\
    \frac{\partial V_j}{\partial \theta_j}= -\frac{1}{N} \sum_{k\neq j}^N \cos(\theta_j - \theta_k) & i = j  \, .
    \end{cases}
\end{equation}
For the fixed point $\bar{\theta}$ set in \eqref{eq:SType2} we obtain a symmetric matrix $\bar A = A(\bar{\theta})$ 
whose entries are
\begin{equation}
\bar A_{ij} =
    \begin{cases}
    \frac{1}{N} &  i \neq j \text{ and } i,j \neq N   \\[2mm]
    -\frac{1}{N} &  i\neq j \mbox{ and } i=N \text{ or } j=N
    \\[2mm]
    -\frac{N-3}{N} &  i=j  < N \\[2mm]
    \frac{N-1}{N} & i=j=N \, .
    \end{cases}
\end{equation}
By the symmetry of $\bar A$, the extremal representation of the eigenvalues $\lbrace \lambda_k \rbrace_{k=1}^N$ of $\bar A$ are  given by the optimization problem:
\begin{equation}\nonumber
    \max_{1\leq k \leq N}  \lambda_k = \max_{\lVert x \rVert=1 } \lbrace x' \bar{A} x \rbrace, \quad  \min_{1\leq k \leq N}  \lambda_k = \min_{\lVert x \rVert=1 } \lbrace x'\bar{A} x \rbrace \, .
\end{equation}
Setting $x$ to be the standard-basis vectors $\mathbf{e}_i$, where $\mathbf{e}_i$ denotes the vector with a 1 in the $i$th coordinate and 0's elsewhere, we see that
\begin{equation}\nonumber
    \min_{1\leq k \leq N}  \lambda_k \leq \min_{1\leq i \leq N} \lbrace \bar{A} \rbrace_{ii}=-\frac{N-3}{N}<0\,,\qquad 0<\frac{N-1}{N} = \max_{1 \leq i \leq N} \lbrace \bar{A} \rbrace_{ii} \leq \max_{1\leq k \leq N}  \lambda_k\,.
\end{equation}
Therefore, there exists at least one positive eigenvalue and at least one negative eigenvalue. Indeed, the matrix $\bar{A}$ has the eigenvalues $\lambda_{-}=-(N-2)/N$ with algebraic multiplicity $N-2$, $\lambda_2=0$ and $\lambda_3=1$
with algebraic multiplicity $1$. This can be checked considering the proposed subspaces of the center, stable and unstable subspace of the linearized system at $\bar \theta$ 
\begin{equation*}
E^c=
    \left\lbrace \begin{bmatrix}
1\\ 
1\\ 
\vdots \\ 
1\\ 
\vdots\\ 
1
\end{bmatrix}
\right\rbrace, \,
E^s= \left\lbrace
\begin{bmatrix} 

-1\\ 
1\\ 
0 \\ 
\vdots\\ 
0\\ 
0
\end{bmatrix}, \begin{bmatrix}
-1\\ 
0\\ 
1 \\ 
0\\ 
\vdots \\ 
0
\end{bmatrix}, \ldots, \begin{bmatrix}
-1\\ 
0\\ 
\vdots \\ 
0\\ 
1 \\ 
0
\end{bmatrix}
\right\rbrace
\text{ and }
E^u=\left\lbrace
\begin{bmatrix}
-1\\ 
-1\\ 
\vdots \\ 
-1\\ 
-1 \\ 
N-1
\end{bmatrix}  \right\rbrace.
\end{equation*}
 Then, the Center Manifold Theorem \cite[p.116]{perko2013differential} yields the existence of an $(N-2)$-dimensional stable manifold $W^s(\bar{\theta})$ tangent to the stable subspace $E^s$, and the existence of a $1$-dimensional unstable manifold $W^u(\bar{\theta})$, and $1$-dimensional center manifold $W^c(\bar{\theta})$ tangents to the $E^u$ and $E^c$ subspaces respectively. Consequently, the dimension of the center manifold conjoint with the stable manifold is smaller than $N$, which implies a null Lebesgue measure in $\mathbb{R}^n$. 
\end{proof}
\end{proposition}

\noindent
Moreover, we have:

\begin{lemma} {\bf (Synchronization):}
Irrespective of the initial condition $\theta \in \mathcal{T}$, the Dissipation Function obeys:
\begin{equation}
    \lim_{t \rightarrow \infty} \Omega^{f_0,V}(S^t \theta) = \begin{cases} K  \left( N -  1 \right) \, , ~~ \mbox{ for } ~ \theta \ne \Theta^\dagger 
    \\ \text{ }\\
    K \left( N- 1\right) \left( \frac{N-4}{N} \right) ~~ \mbox{ for } ~ \theta = \Theta^\dagger 
    \end{cases} \label{LimOmega}
\end{equation}
where $K(N-1)$, the maximum of $\Omega^{f_0,V}$ in $\mathcal{T}^N$,  corresponds to $(N,0)$ synchronization.
\end{lemma}
\begin{proof}
Because of Theorem 2.4 in Ref.\cite{BCM-CMS2015} and of the continuity of $\Omega^{f_0,V}$, the long time limit of $\Omega^{f_0,V} \circ S^t$ 
in the case $\theta \ne \Theta^\dagger$ is given by $\Omega^{f_0,V}(\Theta^*)$. Then,  Eq.\eqref{R2} and Eq.\eqref{omega}, yield
the first line of Eq.\eqref{LimOmega}. The case $\theta = \Theta^\dagger$, gives, instead:
\begin{equation}
    R^* e^{i\varphi^*}=\frac{1}{N} \left( (N-1)e^{i\varphi^*} + e^{i(\varphi^* + \pi)} \right) = \frac{N-2}{N} e^{i\varphi^*}\,.
\end{equation}
Substituting in Eq.\eqref{omega} we obtain the second line of \eqref{LimOmega}.
\end{proof}

\begin{remark}
Equation \eqref{LimOmega} implies that
\begin{equation}
 \lim_{N \rightarrow \infty} \lim_{t \rightarrow \infty} \frac{\Omega^{f_0,V}(S^t \theta)}{N} = K\,. \label{Ksteady}
\end{equation}
In other words, 
the large $t$ limit followed by the large $N$ limit implies that the coupling constant $K$, which drives the synchronization process in the Kuramoto dynamics \eqref{kura1}, equals the average Dissipation per oscillator. For fixed $N$,
synchronization is also evident from the fact that Eq.\eqref{ineq1} must converge to 0, for $\Omega^{f_0,V}$ to become constant.
\end{remark}
 
\noindent
This also implies $R^2(S^t\theta) \to 1$, as $t \to \infty$. It suffices to consider the definition \eqref{omega}
of $\Omega^{f_0,V}$.
\begin{figure}
    \centering
    \includegraphics[width=0.45\textwidth]{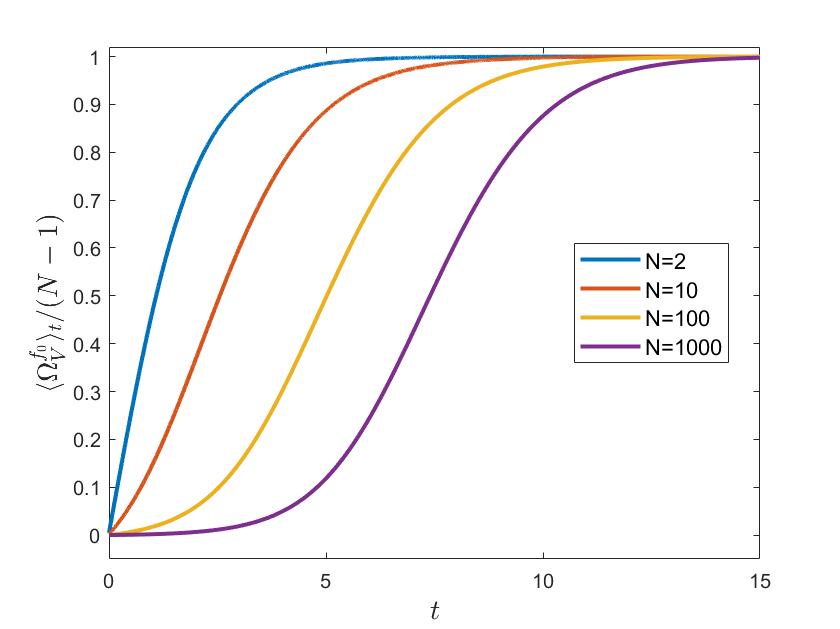}
    \includegraphics[width=0.45\textwidth]{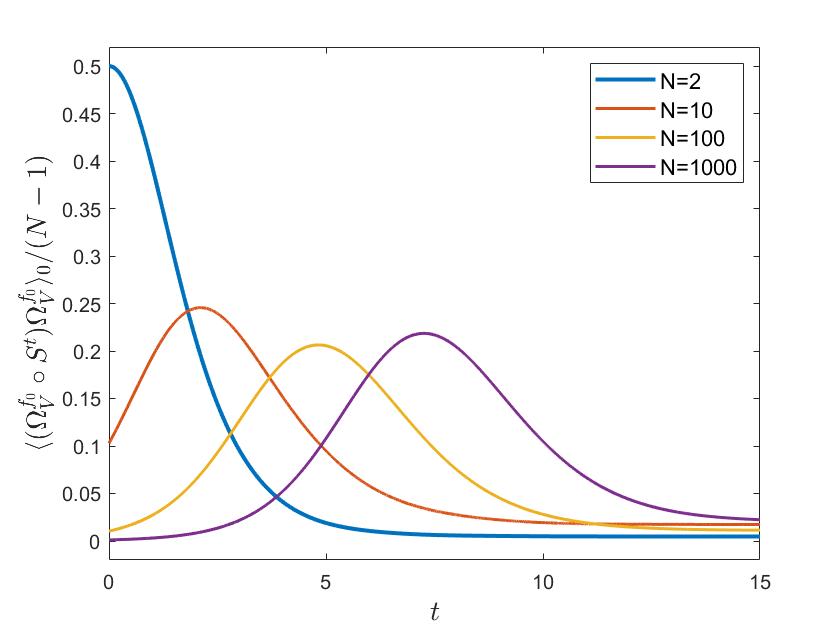}
    \caption{Behavior of $\langle \Omega^{f_0,V} \rangle_t$ (left panel) and $\langle (\Omega^{f_0,V} \circ S^t) \Omega^{f_0,V}\rangle_0$ (right panel), both rescaled by $(N-1)$, as functions of time, for $K=1$, $\omega=0$ and for different values of $N$. The curves on the right panel represent the time derivative of those in the left panel. In particular, $t=0$ in the right panel represents $K^2/N$, cf.\ Eq.\eqref{DOMdt}.}
    \label{fig:fig2}
\end{figure}
For different values of $N$, Fig.~\ref{fig:fig2} illustrates the behavior of $\langle \Omega^{f_0,V} \rangle_t$ and of its time derivative, 
which is $\langle (\Omega^{f_0,V}\circ S^t) \Omega^{f_0,V}\rangle_0$, as functions of time. 
The initial growth of the autocorrelation may look unusual, since autocorrelations are commonly found to decrease. However,
unlike standard calculations that rely on an invariant distribution,\footnote{In linear response the initial distribution is considered invariant to first order in the perturbation.}
our autocorrelation is computed with respect to the transient probability measure $\mu_0$. The figure portrays the result of numerical simulations. 
The right panel of Fig.~\ref{fig:fig2}, shows that for sufficiently large $N$ the autocorrelation function 
$\langle (\Omega^{f_0,V}\circ S^t)\Omega^{f_0,V}\rangle_0$ 
reaches a maximum before it decreases, as required for convergence to a steady state.
An interesting result is the following.
\begin{lemma}
For $N\ge 2$, the derivative of the time dependent average of $\Omega^{f_0,V}$, computed at time $t=0$ obeys:
\begin{equation}
   \left. \frac{d}{dt} \left\langle \Omega^{f_0,V} \right\rangle_t \right|_{t=0}=\left\langle \left(\Omega^{f_0,V} \right)^2\right\rangle_0=K^2 \frac{N-1}{N}\,.
    \label{DOMdt}
\end{equation}
\end{lemma}
\noindent
Note that the derivative of the mean Dissipation Function equals its autocorrelation function, as expressed by Eq.\eqref{ineq2}. 
Therefore, Eq.\eqref{DOMdt} gives the value of this autocorrelation function at $t=0$, as shown in the right panel of Fig.~\ref{fig:fig2}.
\begin{proof}
Using \eqref{ident}, \eqref{ident2} and \eqref{ineq1} we find that
\begin{equation}
    \frac{d}{dt} \left\langle \Omega^{f_0,V} \right\rangle_t = \frac{d}{dt} \left\langle \Omega^{f_0,V} \circ S^t \right\rangle_0
    = 2K^2 \left\langle R^2(S^t \theta ) \sum_{j=1}^N \sin^2 \Big( S^t \theta_j - \Phi\left( S^t \theta \right) \Big) \right\rangle_0\,.
    \label{eq:DOmega}
\end{equation}
Thus, at $t=0$, the integrand of \eqref{eq:DOmega} reads
\begin{align}
    &R^2( \theta) \sum_{j=1} \sin^2(\Phi  -  \theta_j)
    = \frac{1}{N^2} \sum_{j=1}^N \left( \sum_{l=1}^N \sin(\theta_l - \theta_j) \right)^2  \label{eq:IDo}\\
    &\qquad \qquad = \frac{1}{N^2} \sum_{j=1}^N \left[  \sum_{l=1}^N \sin^2(\theta_j - \theta_l)  +   \sum_{l=1}^N \sum_{\substack{k=1 \\ k\neq l}}^N \sin(\theta_l -\theta_j )\sin(\theta_k - \theta_j)  \right]\,. \nonumber 
\end{align}
Furthermore, we have:
\begin{multline}
    \int_{0}^{2\pi} \int_0^{2\pi} \sin( \theta_l - \theta_j) \sin(\theta_k-\theta_j)d\theta_l d\theta_k \\= \int_{0}^{2\pi} \sin(\theta_l-\theta_j) d\theta_l \int_{0}^{2\pi} \sin(\theta_k - \theta_j) d\theta_k =0\,. \label{eq:Int2} 
\end{multline}
Therefore, considering \eqref{eq:IDo} and \eqref{eq:Int2} over \eqref{eq:DOmega} at time $t=0$ we have:
\begin{align*}
    \left. \frac{d}{dt}\langle \Omega^{f_0,V} \rangle_t \right|_{t=0} &= 2K^2 \int_{\mathcal{M}} R^2(\theta) \sum_{j=1}^N \sin^2(\Phi-\theta_j) f_0(\theta) d\theta \\
    &= 2\frac{K^2}{N^2} \frac{1}{(2\pi)^N} \int_{\mathcal{M}} \sum_{j=1}^N \sum_{l=1}^N \sin^2(\theta_j-\theta_l) d\theta \\
    &= 2\frac{K^2}{(2\pi)^2} \frac{N-1}{N} \int_{0}^{2\pi} \int_{0}^{2\pi} \sin^2(\theta_1 - \theta_2) d \theta_1 d \theta_2 \\
    &= K^2 \frac{N-1}{N}\,.
\end{align*}
This completes the proof of \eqref{DOMdt}.
\end{proof}




\Section{Comparison with linear response}
\label{sec:sec4}
 
In this Section we compare the foregoing exact response formalism with the standard linear response \cite{EMnoneq}. Consider a perturbed vector field $V_{\eps}$, defined as
\begin{equation}
    V_{\eps}(\theta)=V_0(\theta)+\eps V_p(\theta)
\end{equation}
where the parameter $\eps$ expresses the strength of the perturbation. 
Following Section \ref{sec:sec3}, we identify $\eps$ with $K$, and define:
\begin{eqnarray}
V_0(\theta)&=&\omega \label{def1}\\
V_{p,j}(\theta)&=& R \sin(\Phi-\theta_j) \, , \quad j=1,...,N \, .
\label{def2}
\end{eqnarray}
Correspondingly, we denote by $S_{\eps}^t$ and $S_{0}^t$ the perturbed and unperturbed flows, respectively. From Eq. (\ref{Omega}), we obtain:
\begin{equation}
    \Omega^{f_0,V_{\eps}}=\Omega^{f_0,V_{0}}+\eps \, \Omega^{f_0,V_p}=\eps \, \Omega^{f_0,V_p} \label{omegaGK} \, .
\end{equation}
In particular, we have:
\begin{equation}
\Omega^{f_0,V_p} =
\frac{1}{N}\sum_{i,j=1}^N \cos{(\theta_j-\theta_i})-1  \, .\label{omp}  
\end{equation}
The last equality in Eq.\eqref{omegaGK} derives from the fact that $\Omega^{f_0,V_0}\equiv 0$ if, as assumed, $f_0$ is invariant under the unperturbed dynamics, cf.\ Eq.\eqref{invf0}.
We may then write the {\em exact} response Eq.\eqref{response} as:
\begin{equation}
\langle \mathcal{O}\rangle_{t,\eps}= \langle \mathcal{O}\rangle_0+\eps\int_0^t \langle \left(\mathcal{O}\circ S_{\eps}^\tau \right)\ 
\Omega^{f_0,V_p}\rangle_0 \ d\tau \label{responseGK}
\end{equation}
where $\mathcal{O}\circ S_{\eps}^t$ denotes the observable $\mathcal{O}$ composed with the perturbed flow. Because this formula is exact, the parameter $\eps$ in it does not need to be small, and it appears both as a factor multiplying the integral and as a subscript indicating the perturbed flow $S_\eps^t$.
Next, using Eq. \eqref{ft}, we can write
\begin{equation}
f_{t}(\theta)=\exp\left\{\eps \int_{-t}^0 
\Omega^{f_0,V_p}(S_{\eps}^\tau\theta) \ d\tau\right\}f_0(\theta) \label{fteps}
\end{equation}
which can be expanded about $\eps=0$, and truncated to first order, to obtain the linear approximation of the evolving probability density: 
\begin{eqnarray}
\bar{f}_{t}(\theta;\eps)&=& f_0(\theta) \left( 1 + \eps \left. 
{\frac d { d \eps}}
\exp\left\{\eps \int_{-t}^0 \Omega^{f_0,V_p}(S_{\eps}^\tau\theta) \ d\tau\right\} \right|_{\eps = 0}
\right) \\
&=&
f_0(\theta) \left( 1 + \eps  \int_{-t}^0 \Omega^{f_0,V_p}(S_{0}^\tau\theta) \ d\tau \right) \nonumber\\
&=& f_0(\theta) \left( 1 +\eps  \int_{0}^t \Omega^{f_0,V_p}(S_{0}^{-\tau}\theta) \ d\tau \right) \, .
\label{fteps2}
\end{eqnarray}
Note that
the expansion in the variable $\eps$ of the exponential in Eq.\eqref{fteps}, requires computing the derivatives with respect to $\eps$ of the time integral in it. This, in turn, requires the derivatives of the Dissipation Function $\Omega^{f_0,V_p}(S_{\eps}^\tau\theta)$, and of the evolved trajectory points $S^\tau_\eps \theta$. Because both the Dissipation Function and the dynamics are smooth on a compact manifold, their derivatives are bounded, and their integral up to any time $t$ computed at $\eps=0$ is also bounded. Multiplied by $\eps$, this integral gives a vanishing contribution to the first derivative of the exponential in Eq.\eqref{fteps}. There only remain the exponential and the integral computed at $\eps=0$, multiplied by the increment $\eps$, which is the brackets in Eq.\eqref{fteps2}.
We then define:
\begin{equation}
\overbar{\langle \mathcal{O}\rangle}_{t,\eps}= \int_{\mathcal{M}} \ \mathcal{O}(\theta)\bar{f}_{t}(\theta;\eps)\ d\theta =
\langle \mathcal{O}\rangle_0+\eps\int_0^t 
\left\langle \mathcal{O}\ \left(\Omega^{f_0,V_p} \circ S_{0}^{-\tau}\right)\right\rangle_0 \ d\tau
\, \label{defGK}
\end{equation}
which is the linear response result.
At the same time, the invariance of the correlation function under time translations of the unperturbed dynamics, which is proven in \ref{app:appB}, yields:
\begin{equation}
\overbar{\langle \mathcal{O}\rangle}_{t,\eps} = \left\langle \mathcal{O} \right\rangle_0 + \eps\int_0^t 
\left\langle \left(\mathcal{O}\circ S_{0}^{\tau}\right)\ \Omega^{f_0,V_p}\right\rangle_0 \ d\tau \, . \label{responseGK2}
\end{equation}
It is interesting to note that, unlike the Green-Kubo formulae, which are obtained from small Hamiltonian perturbations, 
here the perturbation is not Ha\-mil\-tonian. Therefore, we may call \eqref{responseGK2} a \textit{generalized} GK formula. It is worth comparing it with
the exact response formula \eqref{responseGK}, as follows:
\begin{equation}
\langle \mathcal{O}\rangle_{t,\eps}-\overbar{\langle \mathcal{O}\rangle}_{t,\eps}= 
\eps \int_0^t 
\left\langle \Big[ \left(\mathcal{O}\circ S_{\eps}^\tau \right)-\left(\mathcal{O}\circ S_{0}^\tau \right) \Big]  \Omega^{f_0,V_p} \right\rangle_0 \ d\tau \label{responseGK4}
\end{equation}
which shows that the two formulae tend to be the same, in the small $\eps$ limit, as expected. Thanks to the use of the Dissipation Function, their difference lies only in the use of the perturbed rather than the unperturbed flow inside $\mathcal O$.

Let us dwell on the response of two relevant observables, in the case in which $V_0=\omega=0$, hence $S_0^t$ is the identity operator, Id. 
First, taking 
$\mathcal{O}=\Omega^{f_0,V_{\eps}}=\eps \, \Omega^{f_0,V_p}$, we find
\begin{eqnarray}
\langle \Omega^{f_0,V_{\eps}}\rangle_{t,\eps}-\overbar{\langle \Omega^{f_0,V_{\eps}}\rangle}_{t,\eps}&=&\int_0^t \left\langle \left[\left(\Omega^{f_0,V_{\eps}}\circ S_{\eps}^\tau \right)-\left(\Omega^{f_0,V_{\eps}}\circ S_{0}^\tau \right)\right]\ \Omega^{f_0,V_{\eps}} \right\rangle_0 \ d\tau  \nonumber\\
&&\hskip -40pt = \int_0^t \left[ \left\langle \left(\Omega^{f_0,V_{\eps}}\circ S_{\eps}^\tau \right) \Omega^{f_0,V_{\eps}}\right\rangle_0 -\left\langle \left(\Omega^{f_0,V_{\eps}} \right)^2\right\rangle_0 \right] \ d\tau 
\label{responseGK5}
\end{eqnarray}
where we used the identity 
$\left( \Omega^{f_0,V_{\eps}}\circ S_{0}^\tau \right) = \Omega^{f_0,V_{\eps}}$, 
which derives from the fact that $S_0^t=$Id, and which yields, cf.\ Eq.\eqref{DOMdt}:
\begin{equation}
\left\langle \left(\Omega^{f_0,V_{\eps}} \right)^2\right\rangle_0=\eps^2 \frac{N-1}{N} \, . \label{form1} 
\end{equation} 
For $N=2$, we can also use the explicit expression \eqref{omom} for the autocorrelation function:
\begin{equation}
    \left\langle \left(\Omega^{f_0,V_{\eps}}\circ S_{\eps}^\tau \right) \Omega^{f_0,V_{\eps}}\right\rangle_0=\frac{\eps^2}{1+\cosh{(\eps \tau)}} \label{form2}
\end{equation}
which leads to:
\begin{equation}
\langle \Omega^{f_0,V_{\eps}}\rangle_{t,\eps}=\eps \tanh{\left(\frac{\eps t}{2}\right)} \, , \quad
\mbox{and} \quad
\overbar{\langle  \Omega^{f_0,V_{\eps}}\rangle}_{t,\eps}= \frac{\eps^2 t}{2} \label{form4}
\end{equation}
so that 
\begin{equation}
\langle \Omega^{f_0,V_{\eps}}\rangle_{t,\eps}=\overbar{\langle \Omega^{f_0,V_{\eps}}\rangle}_{t,\eps}+o(\eps^2) t \, .\label{lr}
\end{equation}
In other words, for any $\eps > 0$, the difference of the two responses is small at small times, but it diverges linearly as time passes.

As a second instance, let us take $\mathcal{O}=\psi=\theta_1-\theta_2$. From \eqref{omega} and \eqref{R2g} we have:
\begin{equation}
    \Omega^{f_0,V_{\eps}}=2 \eps R^2(\psi)-\eps=\frac{2 \eps}{\tan^2\left(\frac{\psi}{2}\right)+1}-\eps=\eps \cos(\psi) \, .\label{form5}
\end{equation}
Moreover, Eq.\eqref{g} yields:
\begin{equation}
     (\psi\circ S_{\eps}^{t})=2 \arctan{\left[\tan\left(\frac{\psi}{2}\right)e^{-\eps t}\right]} \label{form6}
\end{equation}
and we can write:
\begin{eqnarray}
  \langle \psi\rangle_{t,\eps}-\overbar{\langle \psi\rangle}_{t,\eps} &=& \int_0^t \left[ \left\langle \left(\psi\circ S_{\eps}^\tau \right)\Omega^{f_0,V_{\eps}}\right\rangle_0 -\left\langle \left(\psi\circ S_{0}^\tau \right) \Omega^{f_0,V_{\eps}}\right\rangle_0 \right] \ d\tau \nonumber\\
  &=& \int_0^t \left[ \left\langle \left(\psi\circ S_{\eps}^\tau \right) \Omega^{f_0,V_{\eps}}\right\rangle_0 -\left\langle \psi \  \Omega^{f_0,V_{\eps}}\right\rangle_0 \right] \ d\tau \label{form7}
\end{eqnarray}
where we used $S_0^t=$Id, which implies $\left(\psi\circ S_{0}^\tau \right)\equiv\psi$.
Therefore, using \eqref{form5} and \eqref{form6} in \eqref{form7}, we obtain:
\begin{eqnarray}
    &&\langle \psi\rangle_{t,\eps}-\overbar{\langle \psi\rangle}_{t,\eps} \nonumber \\
    && \qquad = \frac{1}{(2\pi)^2}
    \int_0^t \int_{\mathcal{M}} 2 \arctan{\left[\tan\left(\frac{\theta_1-\theta_2}{2}\right)e^{-\eps \tau}\right]}  \cos (\theta_1 - \theta_2) d\theta d \tau \nonumber\\
    && \qquad  \quad  - \frac{1}{(2\pi)^2} \int_0^t \int_{\mathcal{M}} \left( \theta_1-\theta_2\right)  \cos (\theta_1 - \theta_2) d\theta d\tau = 0 \, .
    \label{form8}
\end{eqnarray}
The last equality follows from the fact that the integrands in Eq. \eqref{form8} are odd continuous and periodic functions, that are integrated over a whole period, so that one actually obtains:
\begin{equation}\label{form9}
\langle \psi\rangle_{t,\eps}=\overbar{\langle \psi\rangle}_{t,\eps}\equiv 0  \quad , \quad \forall\, t>0\,.  
\end{equation}
Clearly, there are observables for which the difference of responses is irrelevant,
since they do not evolve in time,
and others for which the difference is substantial, even under small perturbations. In any event, the exact response characterizes the synchronization transition, while the linear response does not.

\Section{Concluding remarks}
\label{sec:sec5}

We investigated the Kuramoto dynamics for identical oscillators through the statistical mechanics framework of response theory. As a reference
(unperturbed) dynamics we took a system of uncoupled oscillators, with statistical properties given by a factorized $N$-body distribution with
uniform marginal densities. Next, we interpreted the classical Kuramoto mean-field dynamics as a perturbation of the reference one. For any finite
number $N$ of oscillators, we then derived an exact response formula whose validity holds for arbitrarily large perturbations, and we computed, both
analytically and numerically, the asymptotic value of the Dissipation Function. The latter is indeed the main ingredient of the exact response
theory, that has been developed and is well established within the framework of nonequilibrium molecular dynamics \cite{Caruso20,dalron,ESW,Typic,Jepps16}. 
Explicit analytical results are given for $N=2$. We also investigated the two-time autocorrelation function of the Dissipation Function, and
highlighted its non-monotonic behavior for sufficiently large $N$. Finally, we compared the exact response formalism with the linear response
regime. We found that the two responses differ substantially, even for very small perturbations, and that only the exact response describes the transition to synchronized states.

This indicates that the exact response theory, which by definition must be capable of describing even systems undergoing non-equilibrium 
phase transitions, may actually be used in practice. Synchronization phenomena, which are ubiquitous in nature, are indeed of that kind.

\section*{Acknowledgements}

\par\noindent
L.\ R.\ acknowledges partial support from Ministero dell’Istruzione e Ministero dell’Uni\-ver\-si\-tà 
e della Ricerca Grant Dipartimenti di Eccellenza 2018-2022
\par\noindent
(E11G18000350001).

\appendix

\Section{Unstable fixed points for the identical case} 
\label{app:appA}
In this section we show explicitly the existence of unstable points in any neighborhood of a fixed point of $(N-1,1)$ type.

\begin{proposition}
Let $\bar \theta$ be the stationary type $(N-1,1)$ solution set in \eqref{eq:SType2} and $\delta>0$. If $\theta=(\theta_1,\ldots,\theta_N)$ satisfy
\begin{gather} \label{eq:AssumpDelta1}
    \left| \theta_j - \varphi^* \right| \leq \delta^2, \quad  j=\left[1,\ldots,N-1 \right] \\ 
    \theta_N= \varphi^*+\pi+\delta 
\label{eq:AssumpDelta2}
\end{gather}
then there exists a $\delta_0$ such that for any $0<|\delta|< \delta_0 
$ one has:
\begin{equation}\label{eq:R-larger-than}
    R^2(\theta)> \left( \frac{N-2}{N} \right)^2
\end{equation}
and therefore $R(S^t\theta)\to 1$ as $t\to\infty$.
\begin{proof}
From the equation \eqref{R2} we have that
\begin{align*}
    R^2(\theta)- \left( \frac{N-2
    }{N} \right)^2 &= \frac{1}{N^2} \sum_{i,j=1}^N \cos(\theta_i-\theta_j) - \sum_{i,j=1}^{N-1} 1 - 2 \sum_{i=1}^{N-1} 1 +1 \\
    &= \frac{1}{N^2} \left[ \underbrace{\sum_{i,j=1}^{N-1} \left[ \cos(\theta_i-\theta_j)-1 \right]}_{I_1} + \underbrace{2 \sum_{j=1}^{N-1} \left[ \cos(\theta_N - \theta_j) +1 \right] }_{I_2} \right] \\
    &= \frac{1}{N^2}(I_1 + I_2) \, .
\end{align*}
Next, we estimate the lower bounds of $I_1$ and $I_2$. We use the elementary inequality $\frac{x^2}{4}\leq 1- \cos(x) \leq \frac{x^2}{2}$, which is valid for $|x|\le x_0$ where $x_0 \in (0,\frac{5\pi}{6})$.
Then, by using \eqref{eq:AssumpDelta1}, for $I_1$ we get
\begin{equation}
    \begin{split}
    I_1&= \sum_{i,j=1}^{N-1} \cos(\theta_i-\theta_j)-1 \geq  - \frac{1}{2} \sum_{i,j=1}^{N-1}( \theta_i - \theta_j)^2 \\
    &= -\frac{1}{2} \sum_{i,j=1}^N \left[ (\theta_i - \varphi^*) + (\varphi^* -\theta_j) \right]^2 \geq  - 2 \delta^4 (N-1)^2
    \end{split}
 \label{eq:I1}
\end{equation}
if $2\delta^2 \leq x_0$. On the other hand,  
for $I_2$ we first observe that for $1\le j \le N-1$
\begin{align*}
    |\theta_N - \pi - \theta_j| &\le  |\theta_N - \pi - \varphi^*| + |\varphi^*- \theta_j|   \\
    &\le \delta^2 + |\delta| \\
    &\le 2 |\delta|
\end{align*}
if we take $|\delta|\le 1$.

Then we can use the inequality \eqref{eq:AssumpDelta1} to obtain that
\begin{align*}
    \cos(\theta_N - \theta_j) +1 & =1- \cos(\theta_N - \pi - \theta_j) \\
    &\ge \frac 1 4 (\delta^2 + |\delta|)^2 = \frac{1}{4} \delta^2 (1+|\delta|)^2\\
    &\geq \frac{\delta^2}{4}
\end{align*}
where we consider that $\delta^2+|\delta|\le  2 |\delta|\le x_0$. Therefore
\begin{equation}
    I_2 \geq (N-1)  \frac{\delta^2}{2}
    \label{eq:I2}
\end{equation}
and it follows from the equations \eqref{eq:I1} and \eqref{eq:I2} that
\begin{align*}
    R^2(\theta)- \left(\frac{N-2}{N} \right)^2 &\geq \frac 1 {N^2}\left\{ \frac{N-1}{2}  \delta^2 - 2 \delta^4 (N-1)^2\right\} \\
    &=  \delta^2 \frac{N-1}{N^2} \left[\frac{1}{2} - 2\delta^2 (N-1) \right]\\
    & > 0 
\end{align*}
for $\delta^2 < \frac{1}{2(N-1)}$. In summary, 
if we choose $\delta_0=\min\{1, \frac{x_0}2, \frac 1 {2\sqrt{N-1}}\}$, then \eqref{eq:R-larger-than} holds. 

Finally, to prove that $R(S^t\theta)\to 1$ as $t\to+\infty$, we use the fact that the function 
$t \mapsto R(S^t\theta)$ is not decreasing and converges to a value $(N-2k)/N >0$ for some integer $k\ge 0$.

By \eqref{eq:R-larger-than} and the monotonicity we deduce that
$R(S^t\theta)>  \frac{N-2k}{N}$ for all $k\ge 1$ and all $t\ge 0$, and therefore we conclude that, necessarily, the limiting value has $k=0$. The proof is complete.  
\end{proof}
\end{proposition}

\Section{Stationary correlation functions}
\label{app:appB}
Given a vector field $V_0$, let $f_0$ be an invariant probability density under the flow $S^t_0$ generated by $V_0$. With the notation set by Eq.\eqref{notat}, let $\Lambda_{0,t}^{0}$ be the time integral over a trajectory segment, from time $0$ to time $t$, of the phase space volume variation rate $\Lambda^{0}$, which is the divergence of the vector field $V_0$.
Two-time correlation functions between two generic observables $\mathcal{A}, \mathcal{B}: \mathcal{M}\rightarrow \mathbb{R}$, evaluated with the density $f_0$, are invariant under the time translations determined by $S^t_0$.
This can be shown as follows. 
First we note that, proceeding as in Eq. \eqref{omevol}, one finds
\begin{eqnarray}
\Omega_{-t,0}^{f_s,V_0}&=&
\int_{-t}^0 
\Omega^{f_s,V_0}(S_0^{\tau}\theta) d\tau=
-\Lambda_{-t,0}^{0}-\int_{-t}^0 \frac{d}{d\tau}\left(\log f_s(S_0^{\tau}\theta)\right) d\tau\nonumber\\
&=& -\Lambda_{-t,0}^{0}-\log \frac{f_s(\theta)}{f_s(S_0^{-t}\theta)} \, . \label{app1}
\end{eqnarray}
Upon setting $s=0$ in \eqref{app1} and using Eq.\eqref{steady}, we find 
$\left(\Omega^{f_0,V_0}\right)_{-t,0}\equiv 0$, from which we obtain the following useful relation 
\begin{equation}
    f_0(\theta)=\exp\Big\{-\Lambda_{-t,0}^{0}(\theta)\Big\}f_0(S_0^{-t}\theta) \label{useful}
\end{equation}
where the exponential term is related to the Jacobian determinant of the dynamics as  \cite{Jepps16}:
\begin{equation}
   \left|\frac{\partial \left(S_0^{-t}\theta\right)}{\partial \theta}\right| 
      =\exp\Big\{-\Lambda_{-t,0}^{0}(\theta)\Big\} \, . \label{jac2}
\end{equation}
Let us look, next, at time correlation functions of the form 
$$
\langle \left(\mathcal{A}\circ S_{0}^{s+\tau}\right)\ \left(\mathcal{B}\circ S_{0}^{t}\right)\rangle_0=\int_{\mathcal{M}} \mathcal{A}(S_{0}^{s+\tau}\theta)\ \mathcal{B} (S_{0}^{t}\theta) f_0(\theta) d\theta
$$
for any $s,t,\tau \in \mathbb{R}$.
By a change of variables, one finds 
\begin{eqnarray}
\langle \left(\mathcal{A}\circ S_{0}^{s+\tau}\right)\ \left(\mathcal{B}\circ S_{0}^{t}\right)\rangle_0
&=&\int_{\mathcal{M}} \mathcal{A}(S_{0}^{s}\theta)\ \mathcal{B} (S_{0}^{t-\tau}\theta) f_0(S_{0}^{-\tau}\theta) d\left(S_{0}^{-\tau}\theta\right)\nonumber\\
&=& \int_{\mathcal{M}} \mathcal{A}(S_{0}^{s}\theta)\ \mathcal{B} (S_{0}^{t-\tau}\theta) f_0(S_{0}^{-\tau}\theta) \left|\frac{\partial \left(S_{0}^{-\tau}\theta\right)}{\partial \theta}\right| d\theta\nonumber\\
&=& \int_{\mathcal{M}} \mathcal{A}(S_{0}^{s}\theta)\ \mathcal{B} (S_{0}^{t-\tau}\theta)  \exp\left\lbrace-\Lambda_{-\tau,0}^{0}\right\rbrace f_0(S_{0}^{-\tau}\theta) d\theta\nonumber\\
&=& \int_{\mathcal{M}} \mathcal{A}(S_{0}^{s}\theta)\ \mathcal{B} (S_{0}^{t-\tau}\theta) f_0(\theta)  d\theta \nonumber\\
&=& \langle \left(\mathcal{A}\circ S_{0}^{s}\right)\ \left(\mathcal{B}\circ S_{0}^{t-\tau}\right) \rangle_0 \label{invar}
\end{eqnarray}
where we used \eqref{jac2} and, in the last line, the formula \eqref{useful}.


\bibliographystyle{plain}
\bibliography{biblio}

\begin{thebibliography}{10}

\bibitem{RevModPhys.77.137}
J.~Acebr\'on, L.~Bonilla, C.~P\'erez, F.~Ritort, and R.~Spigler.
\newblock {The Kuramoto model: A simple paradigm for synchronization
  phenomena}.
\newblock {\em Rev. Mod. Phys.}, 77:137--185, 2005.

\bibitem{Agar}
G.~S. Agarwal.
\newblock {Fluctuation-Dissipation Theorems for Systems in Non-Thermal
  Equilibrium and Applications}.
\newblock {\em Z. Physik}, 252:25--38, 1972.

\bibitem{ARENAS200893}
A.~Arenas, A.~Díaz-Guilera, J.~Kurths, Y.~Moreno, and C.~Zhou.
\newblock Synchronization in complex networks.
\newblock {\em Physics Reports}, 469(3):93--153, 2008.

\bibitem{Maes}
M.~Baiesi, C.~Maes, and B.~Wynants.
\newblock {Nonequilibrium Linear Response for Markov Dynamics, I: Jump
  Processes and Overdamped Diffusions}.
\newblock {\em J. Stat. Phys.}, 137(5):1094, 2009.

\bibitem{BALMFORTH200021}
N.~Balmforth and R.~Sassi.
\newblock A shocking display of synchrony.
\newblock {\em Physica D: Nonlinear Phenomena}, 143(1):21--55, 2000.

\bibitem{BCM-CMS2015}
D.~Benedetto, E.~Caglioti, and U.~Montemagno.
\newblock On the complete phase synchronization for the {K}uramoto model in the
  mean-field limit.
\newblock {\em Commun. Math. Sci.}, 13(7):1775--1786, 2015.

\bibitem{BDL}
T.~Bodineau, B.~Derrida, and J.~L. Lebowitz.
\newblock A diffusive system driven by a battery or by a smoothly varying
  field.
\newblock {\em J. Stat. Phys.}, 140:648--675, 2010.

\bibitem{Caruso20}
S.~Caruso, C.~Giberti, and L.~Rondoni.
\newblock {Dissipation Function: Nonequilibrium Physics and Dynamical Systems}.
\newblock {\em Entropy}, \textbf{22}:835, 2020.

\bibitem{CHOI2012735}
Y.~Choi, S.~Ha, S.~Jung, and Y.~Kim.
\newblock Asymptotic formation and orbital stability of phase-locked states for
  the {K}uramoto model.
\newblock {\em Physica D: Nonlinear Phenomena}, 241(7):735--754, 2012.

\bibitem{CL14}
M.~Colangeli and V.~Lucarini.
\newblock Elements of a unified framework for response formulae.
\newblock {\em J. Stat. Mech. Theory Exp.}, 2014:P01002, 2014.

\bibitem{CMW11}
M.~Colangeli, C.~Maes, and B.~Wynants.
\newblock A meaningful expansion around detailed balance.
\newblock {\em J. Phys. A}, 44(9):095001, 13, 2011.

\bibitem{CR12}
M.~Colangeli and L.~Rondoni.
\newblock Equilibrium, fluctuation relations and transport for irreversible
  deterministic dynamics.
\newblock {\em Physica D: Nonlinear Phenomena}, \textbf{241}(6):681--691, 2012.

\bibitem{CRV12}
M.~Colangeli, L.~Rondoni, and A.~Vulpiani.
\newblock {Fluctuation-dissipation relation for chaotic non-Ha\-mil\-to\-nian
  systems}.
\newblock {\em J. Stat. Mech. Theory Exp.}, 2012:L04002, 2012.

\bibitem{dalron}
S.~Dal~Cengio and L.~Rondoni.
\newblock Broken versus non-broken time reversal symmetry: irreversibility and
  response.
\newblock {\em Symmetry}, 8(8):Art. 73, 20, 2016.

\bibitem{Derrida}
B.~Derrida.
\newblock Non-equilibrium steady states: fluctuations and large deviations of
  the density and of the current.
\newblock {\em J. Stat. Mech. Theory Exp.}, 2007(7):P07023, 45, 2007.

\bibitem{DF-2018}
H.~Dietert and B.~Fernandez.
\newblock The mathematics of asymptotic stability in the {K}uramoto model.
\newblock {\em Proc. R. Soc. A.}, 474(2220):20180467, 20, 2018.

\bibitem{DX-CMS2013}
J.-G. Dong and X.~Xue.
\newblock Synchronization analysis of {K}uramoto oscillators.
\newblock {\em Commun. Math. Sci.}, 11(2):465--480, 2013.

\bibitem{DORFLER20141539}
F.~Dörfler and F.~Bullo.
\newblock {Synchronization in complex networks of phase oscillators: A survey}.
\newblock {\em Automatica}, 50(6):1539--1564, 2014.

\bibitem{ECM}
D.J. Evans, E.G.D. Cohen, and G.P. Morriss.
\newblock Probability of second law violations in shearing steady flows.
\newblock {\em Phys. Rev. Lett.}, 71:2401, 1993.

\bibitem{EMnoneq}
D.J. Evans and G.~Morriss.
\newblock {\em Statistical Mechanics of Nonequilibrium Liquids}.
\newblock Cambridge University Press, 2008.

\bibitem{ES94}
D.J. Evans and D.J. Searles.
\newblock Equilibrium microstates which generate second law violating steady
  states.
\newblock {\em Phys. Rev. E}, 50:1645--1648, 1994.

\bibitem{ESAdvPhys}
D.J. Evans and D.J. Searles.
\newblock {The Fluctuation Theorem}.
\newblock {\em Advances in Physics}, 51(7):1529--1585, 2002.

\bibitem{ESR2005}
D.J. Evans, D.J. Searles, and L.~Rondoni.
\newblock Application of the {G}allavotti--{C}ohen fluctuation relation to
  thermostated steady states near equilibrium.
\newblock {\em Phys. Rev. E}, 71:056120, 2005.

\bibitem{ESW}
D.J. Evans, D.J. Searles, and S.R. Williams.
\newblock On the fluctuation theorem for the dissipation function and its
  connection with response theory.
\newblock {\em J. Chem. Phys.}, 128(014504), 2008.

\bibitem{Typic}
D.J. Evans, S.R. Williams, D.J. Searles, and L~Rondoni.
\newblock On typicality in nonequilibrium steady states.
\newblock {\em J. Chem. Phys.}, 128(014504), 2016.

\bibitem{fell2011role}
J.~Fell and N.~Axmacher.
\newblock The role of phase synchronization in memory processes.
\newblock {\em Nat. Rev. Neurosci.}, 12(2):105--118, 2011.

\bibitem{GC}
G.~Gallavotti and E.G.D. Cohen.
\newblock {Dynamical ensembles in stationary states}.
\newblock {\em J. Statist. Phys.}, {\textbf{80}}:931--970, 1995.

\bibitem{Glass2001}
L.~Glass.
\newblock {Synchronization and rhythmic processes in physiology}.
\newblock {\em Nature}, 410(6825):277--284, 2001.

\bibitem{gupta2018statistical}
S.~Gupta, A.~Campa, and S.~Ruffo.
\newblock {\em Statistical physics of synchronization}.
\newblock Springer, 2018.

\bibitem{ha2016collective}
S.~Ha, D.~Ko, J.~Park, and X.~Zhang.
\newblock Collective synchronization of classical and quantum oscillators.
\newblock {\em EMS Surv. Math. Sci.}, 3(2):209--267, 2016.

\bibitem{Jepps16}
O.G. Jepps and L.~Rondoni.
\newblock {A dynamical-systems interpretation of the dissipation function,
  T-mixing and their relation to thermodynamic relaxation}.
\newblock {\em J. Phys. A: Math. Theor.}, \textbf{49}:154002, 2016.

\bibitem{jiruska2013synchronization}
P.~Jiruska, M.~De~Curtis, J.~Jefferys, C.~Schevon, S.~Schiff, and K.~Schindler.
\newblock Synchronization and desynchronization in epilepsy: controversies and
  hypotheses.
\newblock {\em J. Physiol.}, 591(4):787--797, 2013.

\bibitem{Kubo}
R.~Kubo.
\newblock {The fluctuation-dissipation theorem}.
\newblock {\em Rep. Prog. Phys.}, {\textbf{29}}:255--284, 1966.

\bibitem{FirstKura}
Y.~Kuramoto.
\newblock Self-entrainment of a population of coupled non-linear oscillators.
\newblock In {\em International {S}ymposium on {M}athematical {P}roblems in
  {T}heoretical {P}hysics}, pages 420--422. Springer-Verlag, 1975.

\bibitem{kuramoto1984chemical}
Y.~Kuramoto.
\newblock {\em Chemical oscillations, waves, and turbulence}.
\newblock Springer Series in Synergetics. Springer-Verlag, Berlin, 1984.

\bibitem{CL12}
V.~Lucarini and M.~Colangeli.
\newblock {Beyond the linear fluctuation-dissipation theorem: the role of
  causality}.
\newblock {\em J. Stat. Mech. Theory Exp.}, 2012:P05013, 2012.

\bibitem{MPRV}
U.B.M. Marconi, A.~Puglisi, L.~Rondoni, and A.~Vulpiani.
\newblock {Fluctuation–dissipation: Response theory in statistical physics}.
\newblock {\em Physics Reports}, {\textbf{461}}:111--195, 2008.

\bibitem{motter2013spontaneous}
A.~Motter, S.~Myers, M.~Anghel, and T.~Nishikawa.
\newblock Spontaneous synchrony in power-grid networks.
\newblock {\em Nature Physics}, 9(3):191--197, 2013.

\bibitem{perko2013differential}
L.~Perko.
\newblock {\em Differential equations and dynamical systems}.
\newblock Springer-Verlag New York, 2006.

\bibitem{pikovsky2001}
A.~Pikovsky, M.~Rosenblum, and J.~Kurths.
\newblock {\em {Synchronization: A Universal Concept in Nonlinear Sciences}}.
\newblock Cambridge University Press, 2001.

\bibitem{Ruelle}
D.~Ruelle.
\newblock {General linear response formula in statistical mechanics, and the
  fluctuation-dissipation theorem far from equilibrium}.
\newblock {\em Physics Letters A}, {\textbf{245}}:220--224, 1998.

\bibitem{SRE2007}
D.J. Searles, L.~Rondoni, and D.J. Evans.
\newblock The steady state fluctuation relation for the dissipation function.
\newblock {\em J. Stat. Phys.}, 128(6):1337--1363, 2007.

\bibitem{singer1999neuronal}
W.~Singer.
\newblock Neuronal synchrony: a versatile code review for the definition of
  relations.
\newblock {\em Neuron}, 24(24):49--64, 1999.

\bibitem{STROGATZ20001}
S.~Strogatz.
\newblock {From Kuramoto to Crawford: exploring the onset of synchronization in
  populations of coupled oscillators}.
\newblock {\em Physica D: Nonlinear Phenomena}, 143(1):1--20, 2000.

\end{thebibliography}

\end{document}